\documentclass[envcountsame, english]{llncs}
\usepackage[T1]{fontenc}
\usepackage[latin9]{inputenc}
\usepackage{hyperref}
\usepackage{verbatim}
\usepackage{amsfonts}
\usepackage{amsmath,amssymb}
\usepackage{multirow}
\usepackage{url}
\usepackage{algorithm}
\usepackage{algpseudocode}
\usepackage{graphicx}
\usepackage{amsfonts}
\usepackage{courier}
\usepackage{fullpage}
\usepackage{epstopdf}
\usepackage[usenames]{color}
\usepackage{bm}
\usepackage{times}
\usepackage{babel}
\usepackage{float}

\newtheorem{obs}[theorem]{Observation}
\newtheorem{prop}[theorem]{Proposition}

\renewcommand{\paragraph}[1]{\noindent {\bf #1}}
\newcommand{\eps}{\epsilon}

\newcommand{\R}{\mathbb{R}}

\newcommand{\parent}{\text{\tt parent}}

\newcommand{\dsum}{\displaystyle\sum}

\newcommand{\mcLA}{{\mathcal L \mathcal A}}

\newcommand{\ignore}[1]{}

\ignore{

\renewcommand{\paragraph}[1]{\vspace{1.5mm}\noindent \textbf{#1}}

\setlength{\topmargin}{-0.1in}
\setlength{\textwidth}{6.5in} 
\setlength{\textheight}{9in}
\setlength{\evensidemargin}{-.1in}
\setlength{\oddsidemargin}{-.1in}

\setcounter{tocdepth}{3}
\setcounter{secnumdepth}{2}

\interfootnotelinepenalty=10000
}

\makeatother

\usepackage{babel}
\makeatother

\begin{document}

\title{A Local Computation Approximation Scheme to Maximum Matching}
\author{ Yishay Mansour\thanks{
Supported in part by  the Google Inter-university center for
Electronic Markets and Auctions, by a grant from the Israel Science
Foundation, by a grant from United States-Israel Binational Science
Foundation (BSF), and the Israeli Centers of Research Excellence
(I-CORE) program (Center No. 4/11).} \and Shai Vardi\thanks{
Supported in part by the Google Europe Fellowship in Game Theory.}}
\institute{School of Computer Science, Tel Aviv University, Tel Aviv
69978, Israel. \email{ mansour,shaivar1@post.tau.ac.il}}

\maketitle
\begin{abstract}
We present a polylogarithmic local computation matching algorithm
which guarantees a $(1-\eps)$-approximation to the maximum matching in graphs of bounded degree.
\end{abstract}

\section{Introduction}

Finding \emph{matchings} -  sets of vertex disjoint edges in a graph - has been an important topic of research for computer scientists for over 50 years. Of particular importance is finding  \emph{maximum} matchings - matchings of maximal cardinality.
Algorithms that find a maximum matching have many applications in computer science; in fact, their usefulness extends far beyond the boundaries of computer science - to disciplines such as economics, biology and chemistry.

The first works on matching were based on unweighted bipartite graphs (representing problems such as matching men and women). Hall's marriage theorem \cite{H35} gives a necessary and sufficient condition for the existence of a perfect matching\footnote{A perfect matching includes all the nodes of a bipartite graph.}. The efficient algorithms for the weighted bipartite matching problem date back to the Hungarian method \cite{K55,M57}.
In this work we focus on maximum matchings in general unweighed graphs.
Berge \cite{B57} proved that a matching is a maximum matching if and only if the graph has no augmenting paths with respect to the matching. Edmonds used augmenting paths to find a maximum matching in his seminal work \cite{E65}, in which he showed that a maximum matching can be found in polynomial time.  Much work on matching been done since (e.g., \cite{H06,HP73,MV80,MS04}).
Our work uses ideas from Hopcroft and Karp's algorithm  for finding maximal matching in bipartite graphs \cite{HP73}, which runs  in time  $O(n^{2.5})$.



Local computation algorithms (LCAs) \cite{RTVX11} consider the
scenario in which we must respond to queries (regarding a feasible
solution) quickly and efficiently, yet we never need the entire
solution at once. The replies to the queries need to be {\em
consistent}; namely, the responses to all possibly queries combine
to a single feasible solution. For example, an LCA for matching in a
graph $G$, receives an {\em edge-query} for an edge $e \in G$ and
replies ``yes'' if and only if $e$ is part of the matching. The
replies to all the possible edge queries define a matching in the
graph.

In this work we present a \emph{local computation approximation scheme}
to maximum matching. Specifically, we present an LCA such that for any $\eps>0$, the edge-query replies  comprise a matching that is a $(1-\eps)$-approximation to the maximum matching.
Our LCA requires $O(\log^3{n})$ space, and with probability at least $1- 1/n^2$, for any edge-query, it runs in time $O(\log^4{n})$.
To the best of our knowledge, this is the first local computation approximation algorithm for a problem which provably does not have an LCA.

\subsubsection{Related work.}
In the distributed setting, Itai and Israeli \cite{II86} showed a
randomized algorithm which computes a maximal matching (which is a
$1/2$-approximation to the maximum matching) and runs in
$O(\log{n})$ time with high probability. This result has been
improved several times since (e.g., \cite{CH03,HKL06}); of
particular relevance is the approximation scheme of Lotker et al.
\cite{LPP08}, which, for every $\eps>0$,  computes a
$(1-\eps)$-approximation to the maximum matching in $O(\log{n})$
time.
Kuhn et al., \cite{KMW06} proved that any distributed algorithm,
randomized or deterministic, requires (in expectation)
$\Omega(\sqrt{\log{n}/\log\log{n}})$ time to compute a
$\Theta(1)$-approximation to the maximum matching, even if the
message size is unbounded.

Rubinfeld et al., \cite{RTVX11} showed how to transform distributed
algorithms to LCAs, and gave LCAs for several problems, including
maximal independent set and hypergraph $2$-coloring.  Unfortunately, their method bounds the running time
of the transformed algorithm exponentially in the running time of the distributed algorithm. Therefore, distributed algorithms for approximate maximum matching cannot be (trivially) transformed to LCAs using their technique.

Query trees model the dependency of queries on the replies to other queries, and
were introduced in the local setting by Nguyen and Onak \cite{NO08}. If a random
permutation of the vertices is generated, and a sequential algorithm
is simulated on this order, the reply to a query on vertex $v$ depends only on the replies to queries on the neighbors of $v$ which come before it in the permutation.
Alon et al., \cite{ARVX11} showed that if the running time of an algorithm is $O(f(n))$, where $f$ is polylogarithmic in $n$, a $1/n^2$ - almost $f(n)$-independent ordering
on the vertices can be generated in time $O(f(n)\log^2{n})$, thus
guaranteeing the polylogarithmic space bound of any such algorithm.
Mansour et al., \cite{MRVX12} showed that the size of the query tree
can be bounded, with high probability, by $O(\log{n})$, for graphs of bounded degree.
They also showed that it is possible to transform many on-line algorithms to LCAs. One of their examples is an LCA for maximal matching, which immediately gives a $1/2$-approximation to the maximum matching.
In a recent work, \cite{MV13}, LCAs were presented for mechanism design problems.
One of their impossibility results shows that
any LCA for maximum matching requires $\Omega(n)$ time.

\section{Notation and Preliminaries}
\label{sect:prel}

\subsection{Graph Theory}

For an undirected graph $G=(V,E)$, a {\em matching} is a subset of edges $M
\subseteq E$ such that no two edges $e_1,e_2\in M$ share a vertex.
We denote by $M^*$ a matching of maximum cardinality.
%
An {\em augmenting path} with respect to a matching $M$ is a simple path
whose endpoints are \emph{free} (i.e., not part of any edge in the
matching $M$), and whose edges alternate between $E\setminus M$ and
$M$.
A set of augmenting paths $P$ is  {\em independent} if no two paths
$p_1,p_2\in P$ share a vertex.

For sets $A$ and $B$, we denote $A \oplus B \overset{def}{=} (A \cup
B) \setminus (A \cap B)$. An important observation regarding
augmenting paths and matchings is the following.

\begin{obs}
If $M$ is a matching and $P$ is an independent set of augmenting
paths, then $M\oplus P$ is a matching of size $|M|+|P|$.
\end{obs}

A vertex $u \in V$ is a \emph{neighbor} of vertex $v \in V$ if $(u,v)
\in E$. Let $N(v)$ denote the set of neighbors of $v$, i.e.,
$N(v)=\{u: (v,u)\in E\}$.
We assume that we have direct access both to $N(v)$ and to
individual edges.


An independent set (IS) is a subset of vertices $W \subseteq V$ with
the property that for any $u,v\in W$ we have $(u,v)\not\in E$,
namely, no two vertices $u,v \in W$ are neighbors in $G$. The IS is
\emph{maximal} (denoted by MIS) if no other vertices can be added to it without
violating the independence property.

\subsection{Local Computation Algorithms}

We use the following model of local computation algorithms
(LCAs)\cite{RTVX11}.\footnote{Our model differs slightly from the
model of \cite{RTVX11} in that their model requires that the LCA
\emph{always} obeys the time and space bounds, and returns an error
with some probability. It is easy to see that any algorithm which
conforms to our model can be modified to conform to the model of
\cite{RTVX11} by forcing it to return an error if the time or space
bound is violated.}
A {\em $(t(n),$ $ s(n), $ $\delta(n))$ - local computation
algorithm} $\mcLA$ for a computational problem is a (randomized)
algorithm which receives an input of size $n$, and a query $x$.
Algorithm $\mcLA$ uses at most $s(n)$ memory, and with probability
at least $1-\delta(n)$, it replies to any query $x$ in time $t(n)$.
The algorithm must be \emph{consistent}, that is, the replies to all of the
possible queries combine to a single feasible solution to the
problem.

\subsection{Query Trees}
Let $G=(V, E)$ be a graph of bounded degree $d$.
A real number $r(v)\in [0,1]$ is assigned independently and uniformly at random
to every vertex $v$ in the graph.
We refer to this random number as the \emph{rank} of $v$.
Each vertex in the graph $G$ holds an input $x(v) \in R$, where the range $R$ is some finite set.
A randomized Boolean function $F$ is defined inductively on
the vertices in the graph such that $F(v)$ is a function of the input $x(v)$ at $v$
as well as the values of $F$ at the neighbors $w$ of $v$ for which $r(w)<r(v)$.

We would like to upper bound  the number of queries that are needed to be made
vertices in the graph in order to compute $F(v_0)$ for any vertex
$v_0 \in G$. We turn to the simpler task of bounding the size of a
certain $d$-regular tree, which is an upper bound on the number of
queries. Consider an infinite $d$-regular tree $\mathcal{T}$ rooted
at $v_{0}$. Each node $w$ in $\mathcal{T}$ is assigned independently
and uniformly at random a distinct real number $r(w)\in [0,1]$. For
every node $w \in \mathcal{T}$ other than $v_{0}$, let $\parent(w)$
denote the parent node of $w$. We grow a (possibly infinite)
subtree $T$ of $\mathcal{T}$ rooted at $v$ as follows: a node $w$ is
in the subtree $T$ if and only if $\parent(w)$ is in $T$ and
$r(w)<r(\parent(w))$ .
 We keep growing $T$ in this manner such that a node $w'\in T$ is a leaf node in $T$
if the ranks of its $d$ children are all larger than $r(w')$.
We call the random tree $T$ constructed in this way a \emph{query tree} and
we denote by $|T|$ the random variable that corresponds to the size of $T$.
Note that $|T|$ is an upper bound on the number of queries. 

If the reply to a query $q$ depends (only) on the replies to a set of queries, $Q$, we call  $Q$ the set of \emph{relevant} queries with respect to $q$.

\subsection{Random Orders}

Let $[n]$ denote the set $\{1,\ldots, n\}$.

A distribution $D:\{0,1\}^{n} \to \R^{\geq 0}$ is \emph{$k$-wise
independent} if, when $D$ is restricted to any index subset $S\subset [n]$
of size at most $k$, the induced distribution over $S$ is the
uniform distribution.
%


A random ordering $D_{\mathbf{r}}$ induces a probability
distribution over permutations of $[n]$. It is said to
\emph{$\eps$-almost $k$-wise independent} if for any subset
$S\subset [n]$ of size at most $k$, the variation distance between
the distribution induced by $D_{\mathbf{r}}$ on $S$ and a uniform
permutation over $S$ is at most $\epsilon$.
We use the following Theorem from \cite{ARVX11}.

\begin{theorem}[\cite{ARVX11}]\label{thm:orderings}
Let $n\geq 2$ be an integer and let $2\leq k \leq n$.
Then there is a construction of $\frac{1}{n^{2}}$-almost $k$-wise
independent random ordering over $[n]$ whose seed length is $O(k\log^{2}n)$.
\end{theorem}

We provide a short, intuitive explanation of the construction.
We can construct $n$ $k$-wise independent random variables $Z =
(z_1, \ldots, z_n)$, using a seed of length $k \log{n}$ (see
\cite{ABI86}). We generate $4 \log{n}$ independent copies of
$k$-wise independent random variables, $Z_1,\ldots Z_{4\log{n}}$.
For $i\in[n]$, taking the $i$-th bit of each $Z_j, 1 \leq j \leq
4\log{n}$ makes for a random variable $r(i) \in \{0,1\}^{4
\log{n}}$, which can be expressed as an integer in $\{0,1,\ldots,
n^4-1\}$. The order is induced by $r$ ($u$ comes before $v$ in the
order if $r(u)<r(v)$). The probability that
there exists $u,v\in[n]$ such that $r(u)=r(v)$ is at most $1/n^2$,
hence the ordering is $1/n^2$-almost $k$-wise independent.

\section{Approximate Maximum Matching}
\label{sect:matching}

We present a local computation  approximation scheme for maximum
matching: We show an LCA that, for any $\eps>0$, computes a maximal
matching which is a $(1-\eps)$-approximation to the maximum
matching.

Our main result is the following theorem:

\begin{theorem}
\label{thm:main}
Let $G=(V,E)$ be a graph of bounded degree $d$.
Then there exists an $ $ $(O(\log^4{n}),$  $O(\log^3{n}), 1/n)$ -
LCA that, for every $\eps>0$,  computes a maximal matching which
is  a $(1 - \eps)$-approximation to the maximum matching.
\end{theorem}

Our algorithm is, in essence, an implementation of the abstract
algorithm of Lotker et al., \cite{LPP08}. Their algorithm, relies on
several interesting results due to Hopcroft and Karp \cite{HP73}.
First, we briefly recount some of these results, as they are
essential for the understanding of our algorithm.

\subsection{Distributed Maximal Matching}

While the main result of Hopcroft and Karp \cite{HP73} is an
improved matching algorithm for bipartite graphs, they show the
following useful lemmas for general graphs.
The first lemma shows that if the current matching has augmenting
paths of length at least $\ell$, then using a maximal set of
augmenting paths of length $\ell$ will result in a matching for
which the shortest augmenting path is strictly longer than $\ell$.
This gives a natural progression for the algorithm.

\begin{lemma}\cite{HP73} \label{lemma:hp1}
Let $G=(V,E)$ be an undirected graph, and let $M$ be some matching in $G$.
If the shortest augmenting path with respect to $M$ has length $\ell$ and
$\Phi$ is a maximal set of independent augmenting paths of length
$\ell$, the shortest augmenting path with respect to $M \oplus \Phi$ has
length strictly greater than $\ell$.
\end{lemma}

The second lemma shows that if there are no short augmenting paths then
the current matching is approximately optimal.

\begin{lemma}\cite{HP73} \label{lemma:hp2}
Let $G=(V,E)$ be an undirected graph. Let $M$ be some matching in
$G$, and let $M^*$ be a maximum matching in $G$. If the shortest
augmenting path with respect to $M$ has length $2k - 1 > 1$ then
$|M| \geq (1-1/k)|M^*|$.
\end{lemma}

Lotker et al.,  \cite{LPP08} gave the following abstract
approximation scheme for maximal matching in the distributed
setting.\footnote{This approach was first used by Hopcroft and Karp in \cite{HP73};
however, they only applied it
efficiently in the bipartite setting.} Start with an empty matching.
In stage $\ell = 1,3,\ldots, 2k-1$, add a maximal independent
collection of augmenting paths of length $\ell$. For $k=\lceil
1/\eps \rceil$, by Lemma \ref{lemma:hp2}, we have that the matching
$M_{\ell}$ is a $(1-\epsilon)$-approximation to the maximum
matching.

In order to find such a collection of augmenting paths of
length $\ell$, we need to define a conflict graph:
\begin{definition}\cite{LPP08} Let $G = (V,E)$ be an undirected graph, let $M \subseteq E$
be a matching, and let $\ell > 0$ be an integer.  The
$\ell$-conflict graph with respect to $M$ in $G$, denoted
$C_M(\ell)$, is defined as follows. The nodes of $C_M(\ell)$ are all
augmenting paths of length  $\ell$, with respect to $M$, and
two nodes in $C_M(\ell)$ are connected by an edge if and only if
their corresponding augmenting paths intersect at a vertex of
$G$.\footnote{Notice that the nodes of the  conflict graph represent
\emph{paths} in $G$. Although it should be clear from the context,
in order to minimize confusion, we refer to a vertex in $G$ by
\emph{vertex}, and to a vertex in the conflict graph by
\emph{node}.}
\end{definition}

We present the abstract distributed algorithm of \cite{LPP08},  {\bf
AbstractDistributedMM}.

\begin{algorithm}
\caption{ - {\bf AbstractDistributedMM} - Abstract distributed algorithm with input $G=(V,E)$ and $\eps>0$}\label{alg1}
\begin{algorithmic}[1]
\State $M_{-1}\gets \emptyset$\Comment{$M_{-1}$ is the empty matching}
\State $k\gets \lceil 1/\eps \rceil$
\For{$\ell\gets 1, 3, \ldots, 2k-1$,} \label{forloop}
    \State Construct the conflict graph $C_{M_{\ell-2}}(\ell)$\label{line:constr}
    \State Let $\mathcal{I}$ be an MIS of $C_{M_{\ell-2}}(\ell)$\label{line:MIS}
    \State Let $\Phi(M_{\ell-2})$ be the union of augmenting paths corresponding to $\mathcal{I}$ \label{line:P}
    \State $M_{\ell}\gets M_{\ell-2} \oplus \Phi(M_{\ell-2})$\Comment{$M_{\ell}$ is matching at the end of phase $\ell$} \label{line:last}
\EndFor
\State Output $M_{\ell}$\Comment{$M_{\ell}$ is a $(1-\frac{1}{k+1})$-approximate maximum matching}
\end{algorithmic}
\end{algorithm}


%

Note that for $M_\ell$, the minimal augmenting path is of length at
least $\ell+2$.
This follows since $\Phi(M_{\ell-2})$ is a maximal independent set of
augmenting paths of length $\ell$. When we add $\Phi(M_{\ell-2})$ to $M_{\ell-2}$,
to get $M_\ell$, by Lemma \ref{lemma:hp1} all the remaining augmenting paths are of length at least
$\ell+2$ (recall that augmenting paths have odd lengths).

Lines \ref{line:constr} - \ref{line:last} do the task of computing
$M_\ell$ as follows: the conflict graph $C_{M_{\ell-2}}(\ell)$ is
constructed and an MIS, $\Phi(M_{\ell-2})$, is found in it. 
$\Phi(M_{\ell-2})$ is then used to augment $M_{\ell-2}$, to give
$M_{\ell}$.

We would like to simulate this algorithm locally. Our main challenge
is to simulate Lines \ref{line:constr} - \ref{line:last} without
explicitly constructing the entire conflict graph $C_{M_{\ell-2}}(\ell)$.
To do this, we will simulate an on-line MIS algorithm.

\subsection{Local Simulation of the On-Line Greedy MIS Algorithm}

In the on-line setting, the vertices arrive in some unknown
order, and {\bf GreedyMIS} operates as follows: Initialize the set $I =
\emptyset$. When a vertex $v$ arrives, {\bf GreedyMIS} checks whether any of
$v$'s neighbors, $N(v)$, is in $I$. If none of them are, $v$ is added to
$I$. Otherwise, $v$ is not in $I$. (The pseudocode for {\bf GreedyMIS} can be found in the full version of the paper.)

In order to simulate {\bf GreedyMIS} locally, we first need to fix
the order (of arrival) of the vertices, $\pi$.
If we know that each query depends on at most $k$ previous queries, we do not need to  explicitly generate the order $\pi$ on
all the vertices (as this would take at least linear time). By
Theorem \ref{thm:orderings}, we can produce a
$\frac{1}{n^2}$-almost-$k$-wise independent random ordering on the
edges, using a seed, $s$, of length $O(k\log^2{n})$.

Technically, this is done as follows.
Let $r$ be a function $r:(v,s)\rightarrow [cn^{4}]$, for some
constant $c$.\footnote{Alternately, we sometimes view $r$ as a
function $r:(v,s)\rightarrow [0,1]$: Let $r'$ be a function
$r':(v,s)\rightarrow [cn^{4}]$, and let $f:[cn^4]\rightarrow [0,1]$
be a function that maps each $x \in [cn^4]-\{1\}$ uniformly at
random to the interval $((x-1)/cn^4, x/cn^4]$, and $f$ maps $1$ uniformly at random to the interval
$[0,1/cn^4]$. Then set $r(v,s) = f(r'(v,s))$.} The vertex order
$\pi$ is determined as follows: vertex $v$ appears before vertex $u$
in the order $\pi$  if $r(v,s)<r(u,s)$. Let $G'=(V', E')$ be the
subgraph of $G$ induced by the vertices $V'\subseteq V$; we denote
by $\pi(G', s)$ the partial order of $\pi$ on $V'$. Note that we
only need to store $s$ in the memory: we can then compute, for any
subset $V'$, the induced order of their arrival.

When simulating {\bf GreedyMIS} on the conflict graph $C_M(\ell) =
(V_{C_M}, E_{C_M})$, we only need a subset of the nodes, $V' \subseteq
V_{C_M}$. Therefore, there is
no need to construct $C_M(\ell)$ entirely; only the relevant
subgraph need be constructed. This is the main observation which
allows us to bound the space and time required by our algorithm.


\subsection{LCA for Maximal Matching}

We present our algorithm for maximal matching -  {\bf LocalMM},
and analyze it.  (The pseudocode for {\bf LocalMM} can be found in 
the full version of the paper.)
In contrast to the distributed algorithm, which runs iteratively,
 {\bf LocalMM} is recursive in nature. In each iteration of
{\bf AbstractDistributedMM}, a maximal matching $M_\ell$, is
computed, where $M_\ell$ has no augmenting path of
length less than $\ell$. We call each such iteration
a \emph{phase}, and there are a total of $k$ phases: $1, 3, \ldots
2k-1$. To find out whether an edge $e \in E$ is in $M_\ell$, we
recursively compute whether it is in $M_{\ell-2}$ and whether it is
in $\Phi(M_{\ell-2})$, a maximal set of augmenting paths of length
$\ell$. We use the following simple observation to determine whether
$e \in M_{\ell}$. The observation follows since $M_{\ell}\gets M_{\ell-2}
\oplus \Phi(M_{\ell-2})$.

\begin{obs}
\label{obs:notboth} $e \in M_{\ell}$ if and only if it is in
either in  $M_{\ell-2}$ or in $\Phi(M_{\ell-2})$, but not in both.
\end{obs}

Recall that {\bf LocalMM} receives an edge $e \in E$ as a query, and
outputs ``yes/no''. To determine whether $e \in M_{2k-1}$, it
therefore suffices to determine, for $\ell=1,3, \ldots, 2k-3$, whether
$e \in M_{\ell}$ and  whether $e \in \Phi(M_{\ell})$.

We will outline our algorithm by tracking a single query. (The initialization parameters will be explained at the
end.) When queried on an
edge $e$, {\bf LocalMM} calls the procedure
\textsc{IsInMatching} with  $e$ and the number of phases
$k$.  For clarity, we sometimes omit some of the parameters from the descriptions of the procedures.\\

\paragraph{Procedure \textsc{IsInMatching}} determines whether an edge
$e$ in in the matching $M_\ell$. To determine
whether $e\in M_{\ell}$, \textsc{IsInMatching} recursively checks whether $e \in
M_{\ell-2}$, by calling \textsc{IsInMatching}$(\ell-2)$, and
whether $e$ is in some path in the MIS $\Phi(M_{\ell-2})$ of
$C_{M_{\ell-2}}(\ell)$. This is done by generating all paths $p$ of
length $\ell$ that include $e$, and calling \textsc{IsPathInMIS}$(p)$ on each.  \textsc{IsPathInMIS}$(p)$ checks whether
$p$ is an augmenting path, and if so,  whether it in the independent set of
augmenting paths.  By
Observation \ref{obs:notboth}, we
can compute whether $e$ is in $M_\ell$ given the output of the calls.\\

\paragraph{Procedure \textsc{IsPathInMIS}} receives a path $p$ and
returns whether the path is in the MIS of augmenting paths of length
$\ell$. The procedure first computes all the relevant  augmenting
paths (relative to $p$) using \textsc{RelevantPaths}.
%
Given the set of relevant paths (represented by nodes) and the
intersection between them (represented by edges) we simulate {\bf
GreedyMIS} on this subgraph. The resulting independent set is a set
of independent augmenting paths. We then just need to check if the path
$p$ is in that set. \\

\paragraph{Procedure \textsc{RelevantPaths}} receives a path $p$ and
returns all the relevant augmenting paths relative to $p$.
The procedure returns the subgraph of $C_{M_{\ell-2}}(\ell), C = (V_C,
E_C)$, which includes $p$ and all the relevant nodes. These are exactly the nodes needed for the
simulation of {\bf GreedyMIS}, given the order induced by seed
$s_\ell$.
The set of augmenting paths $V_C$ is constructed iteratively, by
adding an augmenting path $q$ if it intersects some path $q'\in V_C$
and arrives before it (i.e., $r(q,s_\ell)<r(q',s_\ell)$). In order to determine whether to add  path $q$ to $V_C$, we
need first to test if $q$ is indeed a valid augmenting path, which is done
using \textsc{IsAnAugmentingPath}.\\

\paragraph{Procedure \textsc{IsAnAugmentingPath}} tests if a given
path $p$ is an augmenting path. It is based on the following
observation.

\begin{obs}
\label{obs:cons1} For any graph $G=(V,E)$, let $M$ be a matching in
$G$, and let $p=e_1, e_2, \ldots, e_\ell$ be a path in $G$. Path $p$
is an augmenting path with respect to $M$ if and only if all odd numbered edges
are not in $M$, all even numbered edges are in $M$, and both the
vertices at the ends of $p$ are free.
\end{obs}

Given a path $p$ of length $\ell$, to determine whether $p \in
C_{M_{\ell-2}}(\ell)$, \textsc{IsAnAugmentingPath}$(\ell)$ determines, for each edge in the path,
whether it is in $M_{\ell-2}$, by calling \textsc{IsInMatching}$(\ell-2)$.
It also checks whether the end vertices are free, by calling Procedure
\textsc{IsFree}$(\ell)$, which checks, for each vertex, if any of its
adjacent edges are in $M_{\ell-2}$. From Observation \ref{obs:cons1},
\textsc{IsAnAugmentingPath}$(\ell)$ correctly determines whether $p$ is an
augmenting with respect to $M_{\ell-2}$.

We end by describing the initialization procedure \textsc{Initialize}, which is run only once, during the first query. The procedure sets the number of phases to $\lceil 1/\eps \rceil$. It is important to set a different seed $s_{\ell}$ for each phase $\ell$, since the conflict graphs are unrelated (and even the size of the description of each node, a path of length $\ell$, is different). The lengths of the $k$ seeds, $s_1, s_3, \ldots, s_{2k-1}$, determine our memory requirement.

\subsection{Bounding the Complexity}

In this section we prove Theorem \ref{thm:main}.
We start with the following observation:
\begin{obs}
\label{obs:CG} In any graph $G=(V,E)$ with bounded degree $d$, each
edge $e \in E$ can be part of at most $\ell(d-1)^{\ell-1}$ paths of
length $\ell$. Furthermore, given $e$, it takes at most
$O(\ell(d-1)^{\ell-1})$ time to find all such paths.
\end{obs}

\begin{proof}
Consider a path $p=(e_1, e_2, \ldots, e_{\ell})$ of length $\ell$.
If $p$ includes the edge $e$, then $e$ can be in one of the $\ell$
positions. Given that $e_i=e$, there are at most $d-1$ possibilities
for $e_{i+1}$ and for $e_{i-1}$, which implies at most
$(d-1)^{\ell-1}$ possibilities to complete the path to be of length
$\ell$.\qed
\end{proof}

Observation \ref{obs:CG} yields the following corollary.
\ignore{\begin{corollary} \label{corr:recurs} Procedure
\textsc{IsInMatching} makes at most $\ell\cdot(d-1)^{\ell}$  calls
to Procedure \textsc{IsPathInMIS}.
\end{corollary}}

\begin{corollary}
\label{corr:CGsize}\label{corr:degree}
The $\ell$-conflict graph with respect to any matching $M$ in $G=(V,E)$,
$C_M(\ell)$, consists of at most $\ell(d-1)^{\ell-1}|E| = O(|V|)$
nodes, and has maximal degree at most $d(\ell+1)
\ell(d-1)^{\ell-1}$.
\end{corollary}

\begin{proof}
(For the degree bound.) Each path has length $\ell$, and therefore has $\ell+1$
vertices. Each vertex has degree at most $d$, which implies $d(\ell+1)$ edges. Each
edge is in at most $\ell(d-1)^{\ell-1}$ paths.\qed
\end{proof}

Our main task will be to compute a bound on the number of recursive
calls. First, let us summarize a recursive call. The only
procedure whose runtime depends on the order induced by $s_\ell$ is
\textsc{RelevantPaths}, which depends on the number of vertices $V_C$ (which
is a random variable depending of the seed $s_\ell$). To simplify
the notation we define the random variable $X_\ell=d(\ell+1)
\ell(d-1)^{\ell-1}|V_C|$.  Technically, {\bf GreedyMIS} also depends on
$V_C$, but its running time is dominated by the running time of
\textsc{RelevantPaths}.\\

\begin{tabular}{|c|c|}
    \hline
Calling procedure & Called Procedures \\
    \hline
\textsc{IsInMatching}($\ell$) & $1 \times$ \textsc{IsInMatching}($\ell-2$) and
$\ell(d-1)^{\ell-1} \times$ \textsc{IsPathInMIS}($\ell$)\\
\textsc{IsPathInMIS}($\ell$) & $1\times$ \textsc{RelevantPaths}($\ell$) and $1\times$ {\bf GreedyMIS}\\
\textsc{RelevantPaths}($\ell$) & $X_\ell \times$ \textsc{IsAnAugmentingPath}($\ell)$ \\
\textsc{IsAnAugmentingPath}($\ell)$ &  $\ell\times$
\textsc{IsInMatching}($\ell-2$) and $2\times$ \textsc{IsFree}($\ell$)\\
\textsc{IsFree}($\ell$) & $(d-1)\times$ \textsc{IsInMatching}($\ell-2$)\\
    \hline
\end{tabular}\\

From the table, it is easy to deduce the following proposition.
\begin{proposition}
\label{proposition:silly} \textsc{IsAnAugmentingPath}($\ell)$
generates at most $\ell+2(d-1)$ calls to
\textsc{IsInMatching}($\ell-2$), and
%
therefore at most $(\ell+2d-2)\cdot \ell(d-1)^{\ell-1}$ calls to
\textsc{IsPathInMIS}($\ell-2$).
\end{proposition}

We would like to bound $X_{\ell}$, the number of calls to \textsc{IsAnAugmentingPath}$(\ell)$
during a single execution of \textsc{IsPathInMIS}$(G, p, \ell, S)$.
We require the following theorem, the proof of which appears in Section \ref{sect:improv}.

\begin{theorem}
\label{thm:tree}
For any infinite query tree $T$ with bounded degree $d$, there exists a constant $c$, which depends only on $d$, such that for any large enough $N>0$,
\[
Pr[|T|> N]\leq e^{-cN}.
\]
\end{theorem}

 As a query tree $T$ of bounded degree $D=d(\ell+1)
 \ell(d-1)^{\ell-1}$ is an upper bound to $X_\ell$ (by Corollary \ref{corr:degree}, $D$ is an upper bound on the degree of $C_{M_{\ell-2}}(\ell)$), we have the following corollary to Theorem \ref{thm:tree}.
\begin{corollary}
\label{corr:xl}
There exists an absolute constant $c$, which depends only on $d$, such that for any large enough $N>0$,
\[
Pr[X_{\ell}> N]\leq e^{-c N}.
\]
\end{corollary}


Denote by $f_\ell$ the number of calls to
\textsc{IsAnAugmentingPath}$(\ell)$ during one execution of {\bf
LocalMM}. Let $f = \sum_{\ell=1}^{2k-1} f_{\ell}$.\footnote{For all
even $\ell$, let $f_{\ell}=0$.} The base cases of the recursive
calls {\bf LocalMM} makes are \textsc{IsAnAugmentingPath}$(1)$
(which always returns TRUE). As the execution of each procedure of
{\bf LocalMM} results in at least one call to
\textsc{IsAnAugmentingPath}, $f$ (multiplied by some small constant)
is an upper bound to the total number of computations made by {\bf
LocalMM}.

We state the following proposition, the proof of which appears in Section \ref{sect:improv}.
\begin{prop}
\label{proposition:Z2} Let $W_i$ be a random variable. Let $z_1,
z_2,\ldots z_{W_i}$ be random variables, (some possibly equal to $0$
with probability $1$). Assume that there exist constants $c$ and
$\mu$ such that for all $1 \leq j \leq W_i$, $Pr[z_j \geq \mu N]
\leq e^{-cN}$, for all $N>0$. Then there exist constants $\mu_i$ and
$c'_i$, which depend only on $d$, such that for any $q_i>0$,
\[
Pr[\dsum_{j=1}^{W_i} z_j \geq \mu_i q_i|W_i\leq q_i] \leq e^{-c'_i
q_i}.
\]
\end{prop}

Using Proposition \ref{proposition:Z2}, we prove the following:
\begin{proposition}
\label{proposition:Z1}
 For every $1 \leq \ell \leq 2k-1$, there exist constants $\mu_{\ell}$ and $c_{\ell}$, which depend only on $d$ and $\eps$, such that for any large enough $N>0$
\[
Pr[f_{\ell} > \mu_{\ell} N] \leq e^{-c_{\ell} N}.
\]
\end{proposition}
\begin{proof}
The proof is by induction. For the base of the induction, we have, from  Corollary \ref{corr:xl}, that there exists an absolute constant $c_{2k-1}$, which depends only on $d$, such that for any large enough $N>0$,
$
Pr[X_{2k-1}> N]\leq e^{-c_{2k-1} N}.
$
Assume that the proposition holds for $\ell=2k-1, 2k-3, \ldots \ell$, and we show that it holds for $\ell-2$.

Let $b_{\ell} = (\ell+2d-2)\cdot \ell(d-1)^{\ell-1}$. From Proposition \ref{proposition:silly}, we have that each call to \textsc{IsAnAugmentingPath}$(\ell)$ generates at most  $b_{\ell}$ calls to \textsc{IsPathInMIS}$(\ell-2)$, and hence $b_{\ell}\cdot X_{\ell-2}$ calls to  \textsc{IsAnAugmentingPath}$(\ell-2)$.
From Corollary \ref{corr:xl}, we have that there exists an absolute constant $c$, which depends only on $d$, such that for any large enough $N>0$,
\[
Pr[X_{\ell-2}> N]\leq e^{-c N}.
\]
Setting $W_{\ell} = b_{\ell}f_{\ell}$, $f_{\ell-2}=\sum_{j=1}^{W_i}
z_j$,  $q_i = b_{\ell}\mu_{\ell} y_{\ell}$, and $\mu_i =
\mu_{\ell-2}/b_{\ell}\mu_{\ell}$, and letting $c'_i =
c_\ell'/b_\ell\mu_\ell$ in Proposition \ref{proposition:Z2}
implies the following:
\begin{equation}
\label{eqw}
Pr[f_{\ell-2} > \mu_{\ell-2} y_{\ell}|f_{\ell} \leq \mu_{\ell}y_{\ell}] \leq e^{-c'_{\ell} y_{\ell}}.
\end{equation}
We have
\begin{align}
Pr[f_{\ell-2} > \mu_{\ell-2} N] =& Pr[f_{\ell-2} > \mu_{\ell-2} N|f_{\ell} \leq \mu_{\ell}N]\cdot Pr[f_{\ell} \leq \mu_{\ell} N] \notag \\
&+ Pr[f_{\ell-2} > \mu_{\ell-2} N|f_{\ell} >  \mu_{\ell}N]\cdot Pr[f_{\ell} >  \mu_{\ell}N] \notag \\
\leq& Pr[f_{\ell-2} > \mu_{\ell-2} N|f_{\ell} \leq  \mu_{\ell}N] + Pr[f_{\ell} >  \mu_{\ell}N] \notag \\
\leq& e^{-c'_{\ell} N} +  e^{-c_{\ell} N} \label{eqp} \\
=& e^{-c_{\ell-2} N} ,\notag
\end{align}
where Inequality \ref{eqp} stems from Inequality \ref{eqw} and the induction hypothesis.\qed
\end{proof}
Taking a union bound over all $k$ levels immediately gives
\begin{lemma}
\label{lemma:Z1}
There exists a constant $c$, which depends only on $d$ and $\eps$, such that
\[
Pr[f > c \log{n}] \leq 1/n^2.
\]
\end{lemma}

\begin{proof}[Proof of Theorem \ref{thm:main}]
Using  Lemma \ref{lemma:Z1}, and taking a union bound over all possible queried edges gives us that with probability at least $1-1/n$, {\bf LocalMM} will require at most $O(\log{n})$ queries. Therefore, for each execution of {\bf LocalMM}, we require at most $O(\log{n})$-independence for each conflict graph, and therefore, from Theorem \ref{thm:orderings}, we require $\lceil 1/\eps \rceil$ seeds of length $O(log^3{n})$, which upper bounds the space required by the algorithm. The time required is upper bound by the time required to compute $r(p)$ for all the required nodes in the conflict graphs, which is $O(\log^4{n})$. \qed
\end{proof}

\section{Combinatorial Proofs}
\label{sect:improv}
We want to bound the total number of queries required by Algorithm {\bf LocalMM}.

Let $T$ be a $d$-regular query tree. As in \cite{ARVX11,MRVX12},  we partition the interval [0,1] into $L\geq d+1$ sub-intervals:
$I_i = (1-\frac{i}{L+1}, 1-\frac{i-1}{L+1}]$,
for $i=1,2,\cdots, L$ and $I_{L+1} = [0, \frac{1}{L+1}]$.
We refer to interval $I_i$ as \emph{level} $i$.
A vertex $v \in T$ is said to be on level $i$ if $r(v) \in I_i$.
Assume the worst case, that for the root of the tree, $v_0$, $r(v_0)=1$.
The vertices on level $1$ form a tree $T_1$ rooted at $v_0$.
Denote the number of (sub)trees on level $i$  by $t_i$.
The vertices on level $2$ will form a forest of subtrees $\{T_2^{(1)}, \cdots, T_2^{(t_2)}\}$,
where the total number of subtrees is at most the sum of the number of children of all the vertices in $T_1$.
Similarly, the vertices on level $i>1$ form a forest of subtrees $F_i = \{T_i^{(1)}, \cdots T_i^{(t_{i})}\}$.
Note that all these subtrees $\{T_{i}^{(j)}\}$ are generated independently by the same stochastic process, as the ranks of all  of the nodes in $T$ are i.i.d. random variables.
Denote $f_i = |F_{i}|$, and let $Y_i = \dsum_{j=1}^{i} f_j$. Note that $F_{i+1}$ can consist of at most $Y_i$ subtrees.

We prove the following theorem.

{
\renewcommand{\thetheorem}{\ref{thm:tree}}
\begin{theorem}

For any infinite query tree $T$ with bounded degree $d$, there exists a constant $c$, which depends only on $d$, such that for any large enough $N>0$,
\[
Pr[|T|\geq N]\leq e^{-cN}.
\]
\end{theorem}
\addtocounter{theorem}{-1}
}

We require the following Lemma from \cite{MRVX12}.

\begin{lemma}[\cite{MRVX12}]
\label{lemma:mrvx}
Let $L\geq d+1$ be a fixed integer
and let $T$ be the $d$-regular infinite query tree.
Then for any $1\leq i \leq L$ and $1 \leq j \leq t_i$,
 there is an absolute constant $c$, which depends only on $d$, such that
for all $N>0$,
\[
\Pr[|T_i^{(j)}|\geq N]\leq e^{-cN}.
\]
\end{lemma}

We first prove the following proposition:

\begin{proposition}
\label{prop:Z}
For any infinite query tree $T$ with bounded degree $d$, there exist constants $\mu_1$ and $c_1$, which depend only on $d$, such that for any $1 \leq i \leq L-1$,  and any $y_i>0$,
\[
Pr[f_{i+1} \geq \mu_1 y_i|Y_i = y_i] \leq e^{-c_1 y_i}.
\]
\end{proposition}

\begin{proof}
 Fix $Y_i= y_i$.
Let $\{z_1, z_2,\ldots z_{y_i}\}$ be integers such that $\forall
1\leq i \leq y_i, z_i\geq 0$ and let $x_i = \dsum_{i=1}^{y_i}z_i$.
By Lemma \ref{lemma:mrvx}, the probability that  $F_{i+1}$ consists
exactly  of trees of size $(z_1, z_2,\ldots z_{y_i})$ is at most
$\displaystyle\prod_{i=1}^{y_i}  e^{-cz_i} = e^{-cx_i}$. There are
$\binom{x_i + y_i}{y_i}$ vectors that can realize
$x_i$.\footnote{This can be thought of as $y_i$ separators of $x_i$
elements.} We want to bound $Pr[f_{i+1} = \mu  y_i|Y_i = y_i]$ for
some large enough constant $\mu>0$. Letting  $x_i=\mu
y_i$, we bound it as follows:

 \begin{align*}
Pr[f_{i+1} =x_i|Y_i = y_i] &\leq \binom{x_i + y_i}{y_i}e^{-x_i}\\
 &\leq \left(\frac{e\cdot(x_i+y_i)}{y_i} \right)^{y_i} e^{-cx_i}\\
 &= \left(\frac{e\cdot(\mu y_i + y_i)}{y_i} \right)^{y_i} e^{-c \mu y_{i}}\\
 &= (e\cdot(1+\mu)) ^{y_i} e^{-c \mu y_{i}}\\
 &= e^{y_i(-c \mu  + \ln(1+\mu) + 1)}\\
 &\leq e^{-c' \mu y_i},
 \end{align*}\\
 for some constant $c'>0$. It follows that
\begin{align*}
Pr[f_{i+1} \geq \mu y_i|Y_i=y_i] &\leq \dsum_{k=\mu y_i}^{\infty} e^{-c' k}\\
&\leq e^{-c_1 y_i},
\end{align*}
for some constant $c_1>0$.\qed
\end{proof}

Proposition \ref{prop:Z} immediately implies the following corollary.

\begin{corollary}
\label{corollary:Z} For any infinite query tree $T$ with bounded
degree $d$, there exist constants $\mu$ and $c$, which depend only on
$d$, such that for any $1 \leq i \leq L-1$,  and any $y_i$,
\[
Pr[f_{i+1} \geq \mu y_i| Y_i \leq y_i] \leq e^{-c y_i}.
\]
\end{corollary}

Corollary \ref{corollary:Z}, which is about query
trees,  can be restated as follows: let $W_i = Y_i$, $q_i = y_i$ and
$\sum_{i=1}^{W_i} z_i = f_{i+1}$. Furthermore, let  $c'_i = c_1$ and
$\mu'_i = \mu_1$ for all $i$. This notation yields the following
proposition, which is unrelated to query trees, and which we used in
Section \ref{sect:matching}:

{
\renewcommand{\thetheorem}{\ref{proposition:Z2}}
\begin{prop}
Let $W_i$ be a random variable. Let $z_1, z_2,\ldots z_{W_i}$ be
random variables (some possibly equal to $0$ with probability $1$).
Assume that there exist constants $c$ and $\mu$ such that for all $1
\leq j \leq W_i$, $Pr[z_j \geq \mu N] \leq e^{-cN}$, for all $N>0$.
Then there exist constants $\mu_i$ and $c'_i$, which depend only on
$d$, such that for any $q_i>0$,
\[
Pr[\dsum_{j=1}^{W_i} z_j \geq \mu_i q_i|W_i\leq q_i] \leq e^{-c'_i
q_i}.
\]
\end{prop}
\addtocounter{theorem}{-1}
}

We need one more proposition before we can prove Theorem  \ref{thm:tree}. Notice that $f_1 = |T_1|$.
\begin{proposition}
\label{prop:y} For any infinite query tree $T$ with bounded degree
$d$,  for any $1 \leq i \leq L$, there exist constants $\mu_i$ and
$c_i$, which depend only on $d$, such that for  and any $N>0$,
\[
Pr[f_{i} \geq \mu_i N] \leq e^{-c_i N}.
\]
\end{proposition}

The proof is similar to the proof of Proposition \ref{proposition:Z1}. We include it for completeness.

\begin{proof}
The proof is by induction on the levels $1 \leq i \leq L$, of $T$.

 For the base of the induction, $i=1$, by Lemma \ref{lemma:mrvx}, we have that  there exist some constants $\mu_1$ and $c_1$ such that
\begin{equation*}
 Pr[f_1 \geq \mu_1 N] \leq  e^{-c_1 N},
 \end{equation*}
as $f_1=|T_1|$.

 For the inductive step, we assume that the proposition holds for levels $1,2,\ldots, i-1$, and  show that it holds for level $i$.

 \begin{align}
 Pr[f_i \geq \mu_i N]   =&  Pr[f_i \geq \mu_i N| Y_{i-1} < \mu_{i-1} N]\cdot Pr[Y_{i-1} < \mu_{i-1} N] \notag \\
 &+  Pr[f_i \geq \mu_i N|Y_{i-1} \geq \mu_{i-1} N]\cdot Pr[ Y_{i-1} \geq \mu_{i-1}]  \notag \\
  \leq & Pr[f_i \geq \mu_i N|Y_{i-1} < \mu_{i-1} N] + Pr[ Y_{i-1} \geq \mu_{i-1}]  \notag \\
  \leq & e^{-c N} + e^{-c_{i-1} N} \label{po}\\
  \leq & e^{-c_{i} N} ,\notag
\end{align}
for some constant $c_i$. Inequality \ref{po} stems from Corollary
\ref{corollary:Z} and the inductive hypothesis.\qed
\end{proof}
We are now ready to prove Theorem \ref{thm:tree}.
\begin{proof}[Proof of Theorem \ref{thm:tree}]
 We would like to bound $Pr[|T|=\dsum_{i=1}^{L}f_{i} \geq \mu N]$. From Proposition \ref{prop:y}, we have that for $1 \leq i \leq L$,
 \[
 Pr[f_{i} \geq \mu_i N] \leq e^{-c_i N}.
 \]
 A union bound on the $L$ levels gives the required result. \qed
\end{proof}


\appendix
\section{Pseudocode for Algorithm {\bf GreedyMIS}}
\label{pseudoMIS}

\begin{algorithm}[H]
\caption{ - {\bf GreedyMIS} - On-line MIS algorithm with input $G=(V,E)$ and vertex permutation $\pi$} \label{alg:greedy}
\begin{algorithmic}[1]
\State $I\gets \emptyset$\Comment{$I$ is a set of independent vertices}.
\State Let $\pi = (v_1, v_2, \ldots, v_n)$.
\For{$i=1 $ to  $ n$}
    \If {$\forall u \in N(v_i)$, $u\not\in I$}
    \State $I \leftarrow  I \cup \{v_i\}$.
    \EndIf
\EndFor
\State Output $I$ \Comment{$I$ is an MIS}.
\end{algorithmic}
\end{algorithm}

\section{Pseudocode for Algorithm {\bf LocalMM}}
\label{pseudo}

\begin{algorithm}[H]
\caption{ - {\bf LocalMM} - LCA for MM with input $G=(V,E)$, $e \in E$ and $\eps>0$}\label{alg:Mell}
\label{alg:main}
\begin{algorithmic}[1]
\State Global $\mathcal{S} = \emptyset$ \Comment{$\mathcal{S}$ is the set of seeds}
\Statex
\Procedure{Main}{$G, e, \eps$}
\If {this is the first execution of Algorithm  {\bf LocalMM}}
\State $(\mathcal{S}, k)  \gets$ \textsc{Initialize}$(G, \eps)$
\EndIf
\State Return \textsc{IsInMatching}$(G, e, 2k-1, \mathcal{S})$.
\EndProcedure

\end{algorithmic}
\end{algorithm}

\begin{algorithm}[H]
\caption{Auxiliary procedures}
\begin{algorithmic}[1]

\Procedure{Initialize}{$G, \eps$} \Comment{This is run only at the first execution}
\State $k\gets \lceil 1/\eps \rceil$.
\For {$\ell=1,3, \ldots 2k-1$}
\State Generate a seed $s_{\ell}$ of length $O(\log^3{n})$. \label{line:seeds} \Comment $s_{\ell}$ is a seed for a random ordering $\pi_{\ell}$ on all possible paths of length $\ell$ in $G$.
\EndFor
\State $\mathcal{S} = \displaystyle\bigcup_{\ell} \mathcal{S}_{\ell}$.
\State Return $(\mathcal{S}, k)$.
\EndProcedure 

\Statex
\Procedure {IsInMatching}{$G$, $e$, $\ell$,  $\mathcal{S}$}
\If {$\ell = -1$ } \Comment{The empty matching}
	\State Return false.
\EndIf
\State $b_1$ = \textsc{IsInMatching}($G$, $e$, $\ell-2$,  $\mathcal{S}$).
\State $b_2 =$ false.
\State $P = \{p \in G : e \in p \wedge |p|=\ell$ \}
\For {all $p \in P$} 
	\If {\textsc{IsPathInMIS($G$, $p$, $\ell$, $\mathcal{S}$)}} \label{line:tough}
		\State $b_2 =$ true.
	\EndIf
\EndFor
	\State Return $b_1 \oplus b_2$.
\EndProcedure

\Statex
\Procedure {IsPathInMIS}{$G$, $p$, $\ell$,  $\mathcal{S}$}
\State $C\leftarrow $ \textsc{RelevantPaths}($G$, $p$, $\ell$, $\mathcal{S}$). \Comment{$C$ is a subgraph of $C_{M_{\ell-2}}(\ell)$}
\State $I \gets$ {\bf Greedy MIS} $(C, \pi(C,s_{\ell}))$
\State $b = (v \in I)$
\State Return $b$
\EndProcedure

\Statex
\Procedure {IsFree}{$G$, $v$, $\ell$, $\mathcal{S}$} \Comment{Checks that a vertex is free}
\State IsFreeVertex $=$ true.
\For {all $u \in N(v)$}\Comment{All edges touching $v$}
\If {\textsc{IsInMatching}$(G, (u,v), \ell-2, \mathcal{S})$}
\State IsFreeVertex  $=$ false.
\EndIf
\EndFor
\State Return IsFreeVertex.
\EndProcedure

\end{algorithmic}
\end{algorithm}

\begin{algorithm}
\caption{More auxiliary procedures}
\begin{algorithmic}[1]

\Procedure {RelevantPaths}{$G$,  $p$, $\ell$,  $\mathcal{S}$}
\State Initialize $C=(V_C,E_C) \gets (\emptyset, \emptyset)$.
\If {\textsc{IsAnAugmentingPath}$(G, p, \ell, \mathcal{S})$}
\State $V_C = \{p\}$.
\Else
\State Return $C$.
\EndIf
\While { $\exists p\in V_C: (p,p') \in E_C, r_{\ell}(p',s_{\ell})<r_{\ell}(p, s_{\ell})$}
\If {\textsc{IsAnAugmentingPath}($G,  p', \ell, \mathcal{S}$)}
\State $V_C \leftarrow p'$
\For {all $p'' \in N(p')$} \Comment{Edges between $p'$ and vertices in $V_C$}
\If {$p'' \in V_C$}
\State $E_C \leftarrow (p',p'')$.
\EndIf
\EndFor
\EndIf
\EndWhile
\State Return $C$.
\EndProcedure

\Statex
\Procedure {IsAnAugmentingPath}{$G$, $p$, $\ell$, $\mathcal{S}$} \Comment{Checks that $p$ is an augmenting path.}
\State If $\ell=1$ return TRUE. \Comment{all edges are augmenting paths of the empty matching}
\State Let $p= (e_1, e_2, \ldots, e_\ell)$, with end vertices $v_1, v_{\ell+1}$.
\State IsPath $=$ true.
\For {$i=1$ to $\ell$}
\If {$i\pmod{2} = 0$} \Comment{All even numbered edges should be in the matching}
\If {$\neg$\textsc{IsInMatching}$(G, e_i, \ell-2,  \mathcal{S})$}
\State IsPath  $=$ false.
\EndIf
\EndIf
\If {$i\pmod{2} = 1$} \Comment{No odd numbered edges should be in the matching}
\If {\textsc{IsInMatching}$(G, e_i, \ell-2,  \mathcal{S})$}
\State IsPath  $=$ false.
\EndIf
\EndIf
\EndFor
\If { ($\neg$\textsc{IsFree}$(G, v_1, \ell,  \mathcal{S}) \vee$ $(\neg$\textsc{IsFree}$(G, v_{\ell+1}, \ell,  \mathcal{S})$}
\State  IsPath  $=$ false. \Comment{The vertices at the end should be free}
\EndIf
\State Return IsPath.
\EndProcedure

\end{algorithmic}
\end{algorithm}


\begin{thebibliography}{1}

\bibitem{ABI86}
Noga Alon, L{\'a}szl{\'o} Babai, and Alon Itai.
\newblock A fast and simple randomized algorithm for the maximal independent
  set problem.
\newblock {\em Journal of Algorithms}, 7:567--583, 1986.

\bibitem{ARVX11}
Noga Alon, Ronitt Rubinfeld, Shai Vardi, and Ning Xie.
\newblock Space-efficient local computation algorithms.
\newblock In {\em Proc.\ 22nd {ACM}-{SIAM} {S}ymposium on {D}iscrete
  {A}lgorithms (SODA)}, pages 1132--1139, 2012.

\bibitem{B57}
Claude Berge.
\newblock Two theorems in graph theory.
\newblock {\em Proceedings of the National Academy of Sciences of the United
  States of America}, 43(9):842--844, 1957.

\bibitem{CH03}
Andrzej Czygrinow and Michal Hanckowiak.
\newblock Distributed algorithm for better approximation of the maximum
  matching.
\newblock In {\em COCOON}, pages 242--251, 2003.

\bibitem{E65}
Jack Edmonds.
\newblock Paths, trees, and flowers.
\newblock {\em Canadian Journal of Mathematics}, 17:449--467, 1965.

\bibitem{H35}
Philip Hall.
\newblock On representatives of subsets.
\newblock {\em J. London Math. Soc.}, 10(1):26--30, 1935.

\bibitem{H06}
Nicholas Harvey.
\newblock Algebraic structures and algorithms for matching and matroid
  problems.
\newblock In {\em Proc.\ 47th {A}nnual {IEEE} {S}ymposium on {F}oundations of
  {C}omputer {S}cience (FOCS)}, pages 531--542, 2006.

\bibitem{HKL06}
Jaap-Henk Hoepman, Shay Kutten, and Zvi Lotker.
\newblock Efficient distributed weighted matchings on trees.
\newblock In {\em SIROCCO}, pages 115--129, 2006.

\bibitem{HP73}
John E. Hopcroft and Richard M. Karp.
\newblock An {$N^{5/2}$} algorithm for maximum matchings in bipartite graphs.
\newblock {\em SIAM Journal on Computing}, 2(4):225--231, 1973.

\bibitem{II86}
Amos Israeli and Alon Itai.
\newblock A fast and simple randomized parallel algorithm for maximal matching.
\newblock {\em Inf. Process. Lett.}, 22(2):77--80, 1986.

\bibitem{KMW06}
Fabian Kuhn, Thomas Moscibroda, and Roger Wattenhofer.
\newblock The price of being near-sighted.
\newblock In {\em Proc.\ 17th {ACM}-{SIAM} {S}ymposium on {D}iscrete
  {A}lgorithms (SODA)}, pages 980--989, 2006.

\bibitem{K55}
Harold W. Kuhn.
\newblock The hungarian method for the assignment problem.
\newblock {\em Naval Research Logistics Quarterly}, 2:83--97, 1955.

\bibitem{LPP08}
Zvi Lotker, Boaz Patt-Shamir, and Seth Pettie.
\newblock Improved distributed approximate matching.
\newblock In {\em Proc.\ 20th ACM {S}ymposium on {P}arallel {A}lgorithms and
  {A}rchitectures (SPAA)}, pages 129--136, 2008.

\bibitem{MRVX12}
Yishay Mansour, Aviad Rubinstein, Shai Vardi, and Ning Xie.
\newblock Converting online algorithms to local computation algorithms.
\newblock In {\em Proc.\ 39th {I}nternational {C}olloquium on {A}utomata,
  {L}anguages and {P}rogramming (ICALP)}, pages 653--664, 2012.

\bibitem{MV13}
Yishay Mansour and Shai Vardi.
\newblock Local algorithmic mechanism design.
\newblock Under submission elsewhere.

\bibitem{MV80}
Silvio Micali and Vijay V. Vazirani.
\newblock An {$O(\sqrt{|V|} |E|)$} algorithm for finding maximum matching in
  general graphs.
\newblock In {\em Proc.\ 21st {A}nnual {IEEE} {S}ymposium on {F}oundations of
  {C}omputer {S}cience (FOCS)}, pages 17--27, 1980.

\bibitem{MS04}
Marcin Mucha and Piotr Sankowski.
\newblock Maximum matchings via gaussian elimination.
\newblock In {\em FOCS}, pages 248--255, 2004.

\bibitem{M57}
James Munkres.
\newblock Algorithms for the assignment and transportation problems.
\newblock {\em Journal of the Society for Industrial and Applied Mathematics},
  5(1):32--38, 1957.

\bibitem{NO08}
Huy N. Nguyen and Krzysztof Onak.
\newblock Constant-time approximation algorithms via local improvements.
\newblock In {\em Proc.\ 49th {A}nnual {IEEE} {S}ymposium on {F}oundations of
  {C}omputer {S}cience (FOCS)}, pages 327--336, 2008.

\bibitem{RTVX11}
Ronitt Rubinfeld, Gil Tamir, Shai Vardi, and Ning Xie.
\newblock Fast local computation algorithms.
\newblock In {\em Proc.\ 2nd {S}ymposium on {I}nnovations in {C}omputer
  {S}cience (ICS)}, pages 223--238, 2011.

\end{thebibliography}
\end{document}